\documentclass[11pt]{article}
\usepackage{cite}
\usepackage{mdframed}
\usepackage{epsfig,subfigure}
\usepackage{amssymb}
\usepackage[fleqn]{amsmath}
\usepackage{color}
\usepackage{arydshln}
\usepackage{graphicx}
\usepackage{color}
\usepackage[dvipsnames]{xcolor}

\definecolor{armygreen}{rgb}{0.29, 0.33, 0.13}

\usepackage{amssymb}
\usepackage{amsmath,amsfonts,amssymb}
\usepackage{verbatim}
\usepackage{stfloats}
\usepackage[bookmarks=false]{}
\usepackage{algorithm}
\usepackage{algcompatible}

\newtheorem{theorem}{Theorem}

\newtheorem{Lemma}{Lemma}

\newtheorem{defn}[theorem]{Definition}

\usepackage{booktabs}

\newcommand{\Fb}{\mathsf{F}}

\newcommand{\Fc}{\mathcal{F}}

\newcommand{\Qc}{\mathcal{Q}}

\newcommand{\Xc}{\mathcal{X}}

\newcommand\blfootnote[1]{%
  \begingroup
  \renewcommand\thefootnote{}\footnote{#1}%
  \addtocounter{footnote}{-1}%
  \endgroup
}

\begin{document}

\title{Breaking the MDS-PIR Capacity Barrier via Joint Storage Coding
}
\author{Hua Sun and Chao Tian}

\maketitle

\blfootnote{Hua Sun (email: hua.sun@unt.edu) is with the Department of Electrical Engineering at the University of North Texas. Chao Tian (chao.tian@tamu.edu) is with the Department of Electrical and Computer Engineering at the Texas A\&M University. 
}

\begin{abstract}
The capacity of private information retrieval (PIR) from databases coded using maximum distance separable (MDS) codes has been previously characterized by Banawan and Ulukus, where it was assumed that the messages are encoded and stored separably into the databases. This assumption was also usually taken in other related works in the literature, and this capacity is usually referred to as the MDS-PIR capacity colloquially. In this work, we considered the question if and when this capacity barrier can be broken through joint encoding and storing of the messages. Our main results are two classes of novel code constructions which allow joint encoding as well as the corresponding PIR protocols, which indeed outperform the separate MDS-coded systems. Moreover, we show that a simple but novel expansion technique allows us to generalize these two classes of codes, resulting in a wider range of the cases where this capacity barrier can be broken. 
\end{abstract}

\allowdisplaybreaks
\section{Introduction}

Private information retrieval (PIR) \cite{PIRfirstjournal} has attracted significant attention from researchers in the fields of theoretical computer science, cryptography,  information theory, and coding theory. In the classical PIR model, a user wishes to retrieve one of the $K$ available messages, from $N$ non-communicating databases, each of which has a copy of these $K$ messages. User privacy needs to be preserved during the retrieval process, which requires that the identity of the desired message not be revealed to any single database. 
To accomplish the task efficiently, good codes need to be designed such that the least amount of data should be downloaded. The inverse of the minimum amount of the download data per-bit of desired message is referred to as the capacity of the PIR system. The capacity of the classical PIR system was characterized precisely in a recent work by Sun and Jafar \cite{sun2016capacity}. 

In distributed systems, databases may fail; moreover, each storage node (database) is also constrained on the storage space. Erasure codes can be used to improve both storage efficiency and failure resistance, which motivated the investigation of PIR from data encoded with maximum distance separable (MDS) codes  \cite{shah2014one,freij2017private,banawan2018capacity,Tajeddine_Rouayheb,jingke2017subScienceChina}, with coding parameter $(N,T)$, i.e., the messages can be recovered by accessing any $T$ databases. The capacity of PIR from MDS-coded databases (MDS-PIR) was characterized by Banawan and Ulukus \cite{banawan2018capacity}, which is usually referred to as the MDS-PIR capacity colloquially. 

In all these existing works, the storage code has been designed such that each message is independently encoded and stored into the dababases, and thus can also be recovered individually. In fact, even when the storage codes are not necessarily MDS codes, most existing works on private information retrieval have assumed this separate coding architecture \cite{kumar2019achieving,attia2018capacity,woolsey2019optimal,banawan2019capacity,raviv2018private,lin2018fundamental}, and the only exceptions\footnote{{\color{black}While in this work, we focus exclusively on the metric of PIR capacity {\color{black} for MDS codes}, there is another interesting line of work in coding theory \cite{Fazeli_Vardy_Yaakobi, Rao_Vardy, blackburn2019pir, zhang2019private, skachek2018batch, vajha2017binary} that focuses on a different metric - virtual server rate \cite{blackburn2019pir} and studies how to jointly code the messages such that a minimum number of servers is used to simulate existing PIR protocols {\color{black}  and here as the resulting joint storage code is not required to be MDS, it often turns out to be non-MDS}.}} we are aware of are \cite{chan2015private,sun2018multiround,tian2018shannon}. Though this architecture of separately encoding of each message offers a simple storage solution with good data reliability, it is by no means the only possible MDS storage coding strategy. Instead, the messages can be stored jointly using an MDS code, which could provide the same level of data reliability at the same amount of storage overhead. Motivated by this observation\footnote{This observation was first briefly mentioned as a footnote in \cite{sun2018multiround}, which can be further traced back to a simple code example given in \cite{chan2015private}.}, we ask the following natural question: When can the MDS-PIR capacity barrier, which was established in  \cite{banawan2018capacity} for separately encoding of the messages using an MDS code, be broken, by allowing jointly encoding of the messages using an MDS code? 

In this work, we show that there are many cases, where by jointly encoding and storing the messages, the messages can be protected using an $(N,T)$ MDS code, but retrieved with less data download than the separate coding architecture. In other words, the capacity barrier for separately encoding of the messages can be broken for these cases. More precisely, the mathematical question we ask is under what $(K,N,T)$ parameters, jointly encoding and storing the MDS-coded messages can provide strict PIR retrieval rate improvement; we show that this can be done at least in the following two cases:
\begin{itemize}
\item $(K, N, T) = (2, N, 2)$ and $N \geq 3$;
\item $(K, N, T) = (K, K + 1, K)$ and  $K \geq 2$.
\end{itemize} 
To establish this result, we provide two novel code constructions and PIR protocols which yield strict performance improvement over the strategy of encoding and storing messages separately using an MDS code. Moreover, we show that through a simple but novel code expansion technique, the MDS-PIR capacity barrier can also be broken for the following cases for an arbitrary integer $m\geq 1$:
\begin{itemize}
\item $(K, N, T) = (2, mN, 2m)$ and $N \geq 3$;
\item $(K, N, T) = (K, m(K + 1), mK)$ and  $K \geq 2$.
\end{itemize} 

The rest of the paper is organized as follows. In Section \ref{sec:model}, we provide a precise description of the system model and problem formulation. In Section \ref{sec:k2} and Section \ref{sec:kk}, we provide two novel joint coding storage codes and PIR protocols. In Section \ref{sec:expansion} we present a technique which yields more general classes of the codes which can strictly improve upon separately encoding and storing the messages. Section \ref{sec:conclusion} finally concludes the paper.  

\section{System Model and Problem Formulation}
\label{sec:model}

In this section, we first provide a formal description of the system model, then proceed to pose the problem we seek to answer in this work. A couple of additional remarks to clarify the relation between our system model and those seen in the literature are given at the end of the section. 
\subsection{System Model}

There are a total of $K$ mutually independent messages $W^1, W^2, \ldots, W^{K}$ in the system. Each message is uniformly distributed over $\Xc^L$,  i.e., the set of length-$L$ sequences in the finite alphabet $\Xc$. The messages are MDS-coded and then distributed to $N$ databases, such that from any $T$ databases, the messages can be fully recovered. Since the messages are $(N,T)$ MDS-coded, it is without loss of generality to assume that $L\cdot K=M\cdot T$ for some integer $M$.  

When a user wishes to retrieve a particular message $W^{k^*}$, $N$ queries $Q_{1:N}^{[k^*]}=(Q_1^{[k^*]},\ldots,Q_{N}^{[k^*]})$ are sent to the databases, where $Q_n^{[k^*]}$ is the query for database-$n$. The retrieval needs to be information theoretically private, i.e., any database is not able to infer any knowledge as to which message is being requested. For this purpose, a random key $\Fb$ in the set $\Fc$ is used together with the desired message index $k^*$ to generate the set of queries $Q_{1:N}^{[k^*]}$. Each query $Q_n^{[k^*]}$ belongs to the set of allowed queries for database-$n$, denoted as $\Qc_n$. After receiving query $Q_{n}^{[k^*]}$, database-$n$ responds with an answer $A_n^{[k^*]}$.  Each symbol in the answers from database-$n$ belongs to a finite field $\mathcal{A}_n$, and the answers may have multiple (and different numbers of) symbols. Using the answers $A^{[k^*]}_{1:N}$ from all $N$ databases, together with $\Fb$ and $k^*$, the user then reconstructs $\hat{W}^{k^*}$. We shall refer to such a system as a $(K,N,T)$ MDS-PIR system.

A more rigorous definition of a $(K,N,T)$ system can be specified by a set of coding functions as follows. In the following, we denote the cardinality of a set $\mathcal{B}$ as $|\mathcal{B}|$. 
\begin{defn}\label{def:fun}
A $(K,N,T)$ MDS-PIR code consists of the following coding components: 
\begin{enumerate}
\item A set of MDS encoding functions:
\begin{eqnarray}
\Phi_n :=\Xc^{LK}\rightarrow \Xc^{M}, \, n \in \{1,\ldots,N\}, 
\end{eqnarray}
where each $\Phi_n$ encodes all the messages together into the information to be stored at database-$n$.

\item A set of MDS decoding recovery functions: 
\begin{eqnarray}
\Psi_{\mathcal{T}}: \Xc^{LK} \rightarrow \Xc^{LK}, 
\end{eqnarray}
for each $\mathcal{T}\subseteq  \{1,\ldots,N\}$ such that $|\mathcal{T}|=T$, whose outputs are denoted as $\tilde{W}_{\mathcal{T}}^{1:K}$;
\item{A query function
\begin{eqnarray*}
&\phi_n: \{1,\ldots,K\} \times \Fc \rightarrow \Qc_n, \quad n \in \{1,\ldots,N\},
\end{eqnarray*}
i.e., for retrieving message $W^{k^*}$,  the user sends the query $Q_{n}^{[k^*]} = \phi_n(k^*, \Fb)$ to database-$n$;}
\item{An answer length function
\begin{eqnarray}
\ell_n: \Qc_n \rightarrow \{0, 1, \ldots \},~\quad n \in \{1,\ldots,N\},
\end{eqnarray}
i.e., the length of the answer from each database, a non-negative integer, is a deterministic function
of the query, but not the particular realization of the messages;}
\item{
An answer generating function
\begin{eqnarray}
\phi_n^{(q_n)}: \Xc^{M} \times \Qc_n \rightarrow \mathcal{A}_n^{\ell_n},\quad q_n \in \Qc_n,\, n \in \{1,\ldots,N\},
\end{eqnarray}
i.e., the answer when $q_n=Q_n^{[k^*]}$ is the  query received by database-$n$;}
\item{ 
A reconstruction function
\begin{eqnarray}
\psi: \prod_{n=1}^{N} \mathcal{A}_n^{\ell_n} \times \{1,\ldots,K\} \times \Fc \rightarrow \Xc^{L},
\end{eqnarray}
i.e., after receiving the answers, the user reconstructs the message as $\hat{W}^{k^*} = \psi(A_{1:N}^{[k^*]}, k^*, \Fb)$.
}
\end{enumerate}
These functions satisfy the following three requirements:
\begin{enumerate}
    \item {\bf MDS recoverable:} For any $\mathcal{T}\subseteq \{1,\ldots,N\} $ such that $|\mathcal{T}|=T$, we have $ \tilde{W}_{\mathcal{T}}^{1:K}=W^{1:K}$. 
    \item {\bf Retrieval correctness:} For any $k^* \in \{1,\ldots,K\} $, we have $\hat{W}^{k^*} = W^{k^*}$.
    \item {\bf Privacy:} For every $k, k' \in\{1,\ldots,K\} $, $n \in\{1,\ldots,N\} $ and $q \in \Qc_n$, 
    \begin{eqnarray}
    \mathbf{Pr}(Q_n^{[k]} = q) = \mathbf{Pr}(Q_n^{[k']} = q).
    \end{eqnarray}
\end{enumerate}
\end{defn}

The retrieval rate is defined as 
\begin{eqnarray}
R:=\frac{L\log |\Xc|}{ \sum_{n=1}^{N} \mathbb{E} (\ell_n)\log|\mathcal{A}_n| }.\end{eqnarray}
This is the number of bits of desired message information that can be privately retrieved per bit of downloaded data. The maximum possible retrieval rate is referred to as the capacity of the $(K,N,T)$ system.

\subsection{Separate vs. Joint MDS Storage Codes}

In the general problem definition we have provided above, the MDS encoding functions $\Phi_n$ allow the messages to be jointly encoded. 
{\color{black} For example, suppose we have $K=2$ messages, $N=3$ databases and from any $T=2$ databases, we may decode both messages. A simple jointly encoded MDS storage code is as follows. Each message has $L=2$ bits, denoted as $W^1 = (a_1, a_2), W^2 = (b_1, b_2)$. Each database stores $M = LK/T = 2$ bits, i.e., database-1 stores $(a_1,a_2)$, database-2 stores $(b_1,b_2)$ and database-3 stores $(a_1+b_1, a_2+b_2)$.}
However, in almost all existing works in the literature, e.g., \cite{banawan2018capacity,tajeddine2018private,Sun_Jafar_MDSTPIR,Wang_Skoglund,shah2014one,jingke2017subScienceChina,Ruida2019ISIT}, the messages are encoded separately. In other words, the MDS encoding functions have the special form 
\begin{eqnarray}
\Phi_n = (\Phi_n^1,\Phi_n^2,\ldots,\Phi_n^K),
\end{eqnarray}
where
\begin{eqnarray}
\Phi_n^k: \Xc^{L} \rightarrow \Xc^{M/K}, \, n \in \{1,\ldots,N\}, k \in \{1,\ldots,K\},
\end{eqnarray}
which encodes message $W^k$ to its MDS-coded form to be stored at database-$n$. Correspondingly, the MDS decoding functions have the form
\begin{eqnarray}
\Psi_{\mathcal{T}}= (\Psi_{\mathcal{T}}^1,\Psi_{\mathcal{T}}^2,\ldots,\Psi_{\mathcal{T}}^K),
\end{eqnarray}
where
\begin{eqnarray}
\Psi_{\mathcal{T}}^k: \Xc^{L} \rightarrow \Xc^{L}, \quad k \in \{1,\ldots,K\},
\end{eqnarray}
which decodes message-$k$ from the information regarding $W^k$ stored in the databases in the set $\mathcal{T}$. 
Particularly, since most practical MDS codes are linear, several existing works have directly assumed the MDS encoding functions to be linear, and moreover, the component coding functions $\Phi_n^k$ for different messages $W^k$\rq{}s are the same; see e.g.,  \cite{banawan2018capacity,tajeddine2018private}. In other words, in this class of codes, the encoding function 
$\Phi_n^k$ can be written as the multiplication of the message vector $W^k$ with an $L\times M/K$ encoding matrix $G_n$, whose elements are also in the finite field $\Xc$.
{\color{black} To compare with the jointly encoded MDS storage example above, we consider the same setting where $K=2$ messages, $L=2$ bits per message, $N=3$ servers, and the MDS parameter $T=2$. A separate MDS storage code where each database stores $M/K$ = 1 bit per message is as follows. Database-1 stores $(a_1,b_1)$, database-2 stores $(a_2,b_2)$ and database-3 stores $(a_1+a_2, b_1+b_2)$. It is easy to see that for separately encoded MDS storage codes, the storage space is divided evenly to each message and each divided storage space can only be a function of the corresponding message.}

Let us denote the capacity of $(K,N,T)$ MDS-PIR system as $C(K,N,T)$, that of separate MDS coding as $C_{\perp}(K,N,T)$, and that of separate linear MDS coding with a uniform component function as $C_{\oplus}(K,N,T)$. It is clear from the definitions that 
\begin{eqnarray}
C(K,N,T)\geq C_\perp(K,N,T) \geq C_\oplus(K,N,T).
\end{eqnarray} 
It was shown in \cite{banawan2018capacity} that
\begin{eqnarray}
C_\oplus(K,N,T)=\left(1+ \frac{T}{N}+ \cdots + \left( \frac{T}{N} \right)^{K-1} \right)^{-1}. \label{eqn:MDS_C}
\end{eqnarray}
However, a close inspection of the converse proof in \cite{banawan2018capacity} reveals that 
\begin{eqnarray}
C_\perp(K,N,T)=C_\oplus(K,N,T).
\end{eqnarray}

The issue we thus wish to understand in this work is the relation between $C(K,N,T)$ and $C_\perp(K,N,T)$. In particular, we wish to identify the set of the $(K,N,T)$ triples such that 
\begin{eqnarray}
C(K,N,T)>C_\perp(K,N,T),
\end{eqnarray}
if the set is not empty. We shall show in this work that such triples indeed exist, and they in fact span a rather wide range. 

\subsection{Further Remarks on the System Model}

The result in \cite{banawan2018capacity} is in fact slightly stronger than we have stated in (\ref{eqn:MDS_C}). Let us assume a particular MDS storage code $\mathcal{C}$ is used in the $(K,N,T)$ system, then the corresponding capacities of the $(K,N,T)$ systems as described above can be denoted as $C(K,N,T,\mathcal{C})$, $C_\perp(K,N,T,\mathcal{C})$, and $C_\oplus(K,N,T,\mathcal{C})$, respectively. The result in \cite{banawan2018capacity} can then be stated as that for any linear MDS code $\mathcal{C}$, 
\begin{eqnarray}
C_\oplus(K,N,T,\mathcal{C})=C_\oplus(K,N,T)=\left(1+ \frac{T}{N}+ \cdots + \left( \frac{T}{N} \right)^{K-1} \right)^{-1}.
\end{eqnarray}
It is natural to ask whether for any particular MDS code $\mathcal{C}$, which is not necessarily linear or does not necessarily use a uniform component MDS coding function, whether $C_\perp(K,N,T,\mathcal{C})=C_\perp(K,N,T)$, and more generally whether for any MDS code $\mathcal{C}$, $C(K,N,T,\mathcal{C})=C(K,N,T)$. We believe this is in general not true, however, it appears difficult to prove or disprove this conjecture. 

The MDS recovery requirement implies the following information theoretic relation:
\begin{eqnarray}
\sum_{n\in \mathcal{T}}H(\Phi_n(W^{1:K}))=K L\log|\mathcal{X}|,\\
H(W^{1:K}|\Phi_n(W^{1:K}), n\in\mathcal{T})=0,
\end{eqnarray}
for any $\mathcal{T}\subseteq \{1,2,\ldots,N\}$ and $|\mathcal{T}|=T$. These conditions can be used to derive converse results for a $(K,N,T)$ system, and sometimes are stated directly (e.g. \cite{Sun_Jafar_MDSTPIR}) as the MDS recovery requirement, instead of enforcing the MDS recovery property on the coding functions.

\section{Code Construction: $(K,N,T) = (2,N,2), N \geq 3$} \label{sec:k2}
In this section, we present the storage and PIR code construction when $K=T=2, N \geq 3$ and show that the PIR rate achieved with the proposed joint MDS storage code is strictly higher than the capacity of PIR with separate MDS storage code, i.e., $C(2,N,2) > C_\perp(2,N,2)$.
\subsection{Example: $N=4$}\label{sec:ex4}
To illustrate the main idea in a simpler setting, we start with an example where $N=4$. We set message size $L = 3$ so that each message consists of $3$ symbols from $\mathbb{F}_3$. Denote $W^1 = (a_0; a_1; a_2) \in \mathbb{F}_3^{3\times 1}, W^2 = (b_0; b_1; b_2) \in \mathbb{F}_3^{3 \times 1}$.

{\bf Storage Code:} From the joint MDS storage code constraint, each database stores $\frac{LK}{T} = 3$ symbols, and the stored variables are specified in the following table.
\vspace{-0.15in}

\begin{table}[H]
\caption{Stored Variables.}
\vspace{0.05in}
\centering
\begin{tabular}{cccc}
\toprule
\textbf{Database-1}	& \textbf{Database-2}	& \textbf{Database-3} & \textbf{Database-4} \\
\midrule
$a_0$		& $b_0$			& $a_1+b_0$		&$2a_2+b_0$\\
$a_1$		& $b_1$			& $a_2+b_1$		&$2a_0+b_1$\\
$a_2$		& $b_2$			& $a_0+b_2$		&$2a_1+b_2$\\
\bottomrule
\end{tabular}
\end{table}

It is easy to verify that we may recover both messages from the storage of any 2 databases. For example, consider database-3 and database-4. It suffices to show that $(a_1 - 2a_2; a_2 - 2a_0; a_0 - 2a_1)$ are invertible to $W^1 = (a_0; a_1 ; a_2)$. Equivalently, we show that the following matrix has full rank over $\mathbb{F}_3$.
\begin{eqnarray}
\left[\begin{array}{ccc}
0 & 1 & -2 \\
-2 & 0 & 1 \\
1 & -2 & 0
\end{array}
\right] \rightarrow 
\left[\begin{array}{ccc}
0 & 1 & 1 \\
1 & 0 & 1 \\
1 & 1 & 0
\end{array}
\right] \rightarrow 
\det \left[\begin{array}{ccc}
0 & 1 & 1 \\
1 & 0 & 1 \\
1 & 1 & 0
\end{array}
\right]
= 2 \neq 0
\end{eqnarray}

{\bf PIR Code:} When we retrieve $W^1$, the answers are shown in the following table.
\vspace{-0.15in}

\begin{table}[H]
\caption{Answers for $W^1$.}
\vspace{0.05in}
\centering
\begin{tabular}{ccccc}
\toprule
$\mathsf{F}$ & \textbf{Database-1}	& \textbf{Database-2}	& \textbf{Database-3} & \textbf{Database-4} \\
\midrule
0& $a_0$		& $b_0$			& $a_1+b_0$		&$2a_2+b_0$\\
1& $a_1$		& $b_1$			& $a_2+b_1$		&$2a_0+b_1$\\
2& $a_2$		& $b_2$			& $a_0+b_2$		&$2a_1+b_2$\\
\bottomrule
\end{tabular}
\end{table}

When we retrieve $W^2$, the answers are shown in the following table.
\vspace{-0.15in}

\begin{table}[H]
\caption{Answers for $W^2$.}
\vspace{0.05in}
\centering
\begin{tabular}{ccccc}
\toprule
$\mathsf{F}$ & \textbf{Database-1}	& \textbf{Database-2}	& \textbf{Database-3} & \textbf{Database-4} \\
\midrule
0& $a_0$		& $b_0$			& $a_0+b_2$		&$2a_0+b_1$\\
1& $a_1$		& $b_1$			& $a_1+b_0$		&$2a_1+b_2$\\
2& $a_2$		& $b_2$			& $a_2+b_1$		&$2a_2+b_0$\\
\bottomrule
\end{tabular}
\end{table}

{\bf Correctness and Privacy:} Both correctness and privacy are easy to verify. Correctness follows from the observation that from the 4 symbols downloaded (one from each database), we may decode the 3 desired symbols as only 1 undesired symbol appears in the answers. Privacy is guaranteed because no matter which message is desired, for each database, the answers are identically distributed. For example, consider database-3. The answers are equally likely to be $a_0+b_2, a_1+b_0$ and $a_2+b_1$, regardless of the desired message index.

{\bf Rate that outperforms separate MDS-PIR capacity:} The desired message has $L=3$ symbols and we are downloading one symbol from each database, $l_n = 1, \forall n \in \{1,2,3,4\}$. Then the rate achieved is $\frac{L}{\sum_n l_n} = \frac{3}{4} \leq C(2,4,2)$, which is strictly higher than $C_\perp{(2,4,2)} = (1+\frac{2}{4})^{-1} = \frac{2}{3}$, the capacity of separate MDS storage code.

\subsection{General Proof: Arbitrary $N \geq 3$}
We set message size $L = N-1$, then each message consists of $N-1$ symbols from $\mathbb{F}_{p^m}$ for a prime number $p$ and an integer $m$ such that $p^m \geq (N-3)(N-1)+2$. The primitive element of the finite filed $\mathbb{F}_{p^m}$ is denoted as $\alpha$.
Denote $W^1 = (a_0; a_1; \cdots; a_{N-2}) \in \mathbb{F}_{p^m}^{(N-1) \times 1}, W^2 = (b_0;b_1; \cdots; b_{N-2})\in \mathbb{F}_{p^m}^{(N-1) \times 1}$.

{\bf Storage Code:} From the joint MDS storage code constraint, each database stores $\frac{LK}{T} = N-1$ symbols, and the stored variables $S_n \in \mathbb{F}_{p^m}^{(N-1) \times 1}, n \in \{1,\cdots,N\}$ are set as follows.

Denote the cyclicly shifted message vector as $\tilde{W}^1(i) = (a_{\overline{i}}; a_{\overline{i+1}}; \cdots; a_{\overline{i+N-2}}), i \in \{1,\cdots,N-2\}$ where $\overline{i} = i \mod (N-1)$, i.e., the symbol indices are interpreted modulo $N-1$.
\begin{eqnarray}
S_1 &=& W^1 \\
S_2 &=& W^2 \\
S_3 &=& \alpha \tilde{W}^1(1) + W^2 \\
&\vdots& \\
S_n &=& \alpha^{n-2} \tilde{W}^1(n-2) + W^2\\
&\vdots& \\
S_N &=& \alpha^{N-2} \tilde{W}^1(N-2) + W^2
\end{eqnarray}
Specifically,
\begin{eqnarray}
&& S_1 = (S_{1,0}; \cdots; S_{1,N-2}) = (a_0; \cdots; a_{N-2}) \\
&& S_2 = (S_{2,0}; \cdots; S_{2,N-2}) = (b_0; \cdots; b_{N-2}) \\
&& S_n = (S_{n,0}; \cdots; S_{n,N-2}) = (\alpha^{n-2} a_{\overline{n-2}} + b_{0} ; \cdots; \alpha^{n-2} a_{\overline{n+N-4}} + b_{N-2}) , n \in \{3,\cdots,N\}
\end{eqnarray}

The proof that the above storage code satisfies the MDS criterion is deferred to Section \ref{sec:mdsproof}.

{\bf PIR Code:} When we retrieve $W^1$, the answers are set as follows. $\mathsf{F}$ is uniformly distributed over $\{0,1,\cdots,N-2\}$. When $\mathsf{F} = f \in \{0,1,\cdots, N-2\}$, we set
\begin{eqnarray}
 && A_1^{[1]} = S_{1,f}  = a_f \\
&& A_2^{[1]} = S_{2,f} = b_f \\
&& A_3^{[1]} = S_{3,f} = \alpha a_{\overline{f+1}} + b_f \\
&& ~~~~~~\vdots \\
&& A_n^{[1]} =  S_{n,f} =  \alpha^{n-2} a_{\overline{f+n-2}} + b_f \\
&& ~~~~~~\vdots \\
&& A_N^{[1]} = S_{N,f} = \alpha^{N-2} a_{\overline{f+N-2}} + b_f 
\end{eqnarray}

When we retrieve $W^2$, the answers are set as follows. $\mathsf{F}$ is uniformly distributed over $\{0,1,\cdots,N-2\}$. When $\mathsf{F} = f \in \{0,1,\cdots, N-2\}$, we set
\begin{eqnarray}
 && A_1^{[2]} = S_{1,f} = a_f \\
&& A_2^{[2]} = S_{2,f} = b_f \\
&& A_3^{[2]} = S_{3,\overline{f-1}} = \alpha a_{f} + b_{\overline{f-1}} \\
&& ~~~~~~\vdots \\
&& A_n^{[2]} = S_{n,\overline{f-(n-2)}} =  \alpha^{n-2} a_{f} + b_{\overline{f-(n-2)}} \\
&& ~~~~~~\vdots \\
&& A_N^{[2]} = S_{N,\overline{f-(N-2)}} = \alpha^{N-2} a_{f} + b_{\overline{f-(N-2)}} 
\end{eqnarray}

{\bf Correctness and Privacy:} Similar to the example presented in the previous section, both correctness and privacy are easy to verify. Correctness follows from the observation that the $N$ symbols downloaded (one from each database) contain all $N-1$ desired symbols and only 1 undesired symbol. Specifically, when $W^1$ is desired, we may recover $W^1$ from $(A_1^{[1]}, A_3^{[1]} - A_2^{[1]}, \cdots, A_N^{[1]} - A_2^{[1]})$ and when $W^2$ is desired, we may recover $W^2$ from $(A_2^{[2]}, A_3^{[2]} - \alpha A_1^{[2]}, \cdots, A_N^{[2]} - \alpha^{N-2} A_1^{[2]})$. Privacy is guaranteed because no matter which message is desired, $A_n^{[1]}$ and $A_n^{[2]}$ are identically distributed. For $n=1, 2$, this is trivial to see; when $n\geq 3$, since $A_n^{[1]} = S_{n,f}$, $A_n^{[2]} = S_{n,\overline{f-(n-2)}}$ and $f \in \{0,1,\cdots, N-2\}$, it is seen that $f$ and $\overline{f-(n-2)} = (f - (n-2)) \mod (N-2)$ take values from the same set $\{0,1,\cdots,N-2\}$ for any $n$, and moreover the queries follow the same uniform distribution on this set for both messages. 

{\bf Rate that outperforms separate MDS-PIR capacity:} The desired message has $L=N-1$ symbols and we are downloading one symbol from each database, $l_n = 1, \forall n \in \{1,\cdots,N\}$. Then the rate achieved is $\frac{L}{\sum_n l_n} = \frac{N-1}{N} \leq C(2,N,2)$. When $N \geq 3$,  $C(2,N,2) \geq \frac{N-1}{N} > \frac{N}{N+2} = C_\perp (2,N,2)$, the capacity of separate MDS storage code.

\subsubsection{Proof of MDS storage criterion} \label{sec:mdsproof}
We show that from the stored variables of any two databases, $S_i, S_j, i < j, i, j \in \{1,\cdots, N\}$ we may recover both $W^1$ and $W^2$.

When $i = 1, 2$, the proof is immediate. Henceforth we consider $i \geq 3$. To show that from $(S_i, S_j)$ we may recover $(W^1, W^2)$, it suffices to prove that from $S_i - S_j$, we may recover $W^1$. Note that
\begin{eqnarray}
S_i - S_j &=&  \alpha^{i-2} \tilde{W}^1(i-2)  -  \alpha^{j-2} \tilde{W}^1(j-2) \\
&=& (\alpha^{i-2}a_{\overline{i-2}} - \alpha^{j-2}a_{\overline{j-2}}; \cdots; \alpha^{i-2}a_{\overline{i+N-4}} - \alpha^{j-2}a_{\overline{j+N-4}}) \\
&=& {\bf C}_{i,j} (a_0; \cdots; a_{N-2})
\end{eqnarray}
where ${\bf C}_{i,j}$ is an $(N-1) \times (N-1)$ circulant matrix whose rows consist of all possible cyclic shifts of the following $1\times(N-1)$ row vector,
\begin{eqnarray}
{\bf c} = (c_0, c_1, \cdots, c_{N-2}) = (\alpha^{i-2}, \underbrace{0, \cdots, 0}_{j-i-1~0's}, -\alpha^{j-2}, 0, \cdots, 0).
\end{eqnarray} 
We are left to prove the circulant matrix ${\bf C}_{i,j}$ has full rank. From a result by Ingleton \cite{ingleton1956rank}, a circulant matrix has full rank if the following two polynomials have no common root.
\begin{eqnarray}
f(x) &=& c_0 + c_1 x + \cdots c_{N-2} x^{N-2} = \alpha^{i-2} - \alpha^{j-2} x^{j-i},\\
g(x) &=& x^{N-1} - 1.
\end{eqnarray}
To show that $f(x), g(x)$ have no common root for all integers $i,j, 3\leq i < j \leq N$, we prove by contradiction. Suppose on the contrary that there exists an element $x_0 \in \mathbb{F}_{p^m}$ and two integers $i,j, 3\leq i < j \leq N$ such that $f(x_0) = 0$ and $g(x_0) = 0$, i.e.,
\begin{eqnarray}
\alpha^{i-2} &=& \alpha^{j-2} {x_0}^{j-i} \label{eq:p1}\\
x_0^{N-1}&=& 1 \label{eq:p2}
\end{eqnarray}
Taking (\ref{eq:p1}) to the $(N-1)$-th power, we have
\begin{eqnarray}
\alpha^{(i-2)(N-1)} &=& \alpha^{(j-2)(N-1)} ({x_0}^{N-1})^{j-i} \\
\overset{(\ref{eq:p2})}{\Rightarrow} \alpha^{(i-2)(N-1)} &=& \alpha^{(j-2)(N-1)}\\
\Rightarrow~~~~~~~~~~~~~~ 1 &=& \alpha^{(j-i)(N-1)} \label{eq:p3}
\end{eqnarray}
Note that $(j - i)(N-1) \leq (N-3)(N-1)$. Combining with the assumption that $p^m -2 \geq (N-3)(N-1)$ and $\alpha$ is a primitive element of $\mathbb{F}_{p^m}$, we have \cite{lidl1994introduction}
\begin{eqnarray}
1 \notin \{\alpha, \alpha^2, \alpha^3, \cdots, \alpha^{p^m-2} \} \Rightarrow
\alpha^{(j-i)(N-1)} \neq 1,
\end{eqnarray}
which contradicts (\ref{eq:p3}). The proof is now complete.

{\it Remark: The field size may be further reduced by a result from \cite{guan1986exact}. To ensure ${\bf C}_{i,j}$ has full rank, it suffices to ensure $f(x)$ and $g'(x) = x^r -1$ has no common root, where $N-1 = r p^l$ and $p,r$ are co-prime \cite{guan1986exact}. Using this result and following similar proof steps as above, we may set $p^m \geq (N-3) r + 2$. Note that here $r$ depends on $p$, so to find the smallest field size, we may search by first fixing $p$.}

\section{Code Construction: $(K,N,T) = (K, K+1, K), K \geq 2$}\label{sec:kk}
In this section, we present the storage and PIR code construction when $N=K+1=T+1$ and show that the PIR rate achieved with the proposed joint MDS storage code is strictly higher than the capacity of PIR with separate MDS storage code, i.e., $C(K, K+1, K) > C_\perp(K, K+1, K)$.
\subsection{Example: $(K,N,T)=(3,4,3)$}
To illustrate the main idea in a simpler setting, we start with an example where $K=3, N=4, T=3$. We set message size $L = 2$ so that each message consists of $2$ bits from $\mathbb{F}_2$. Denote $W^1 = (a_1; a_2), W^2 = (b_1; b_2), W^3 = (c_1;c_2)$.

{\bf Storage Code:} From the joint MDS storage code constraint, each database stores $\frac{LK}{T} = 2$ bits, and the stored variables are specified in the following table.
\vspace{-0.15in}

\begin{table}[H]
\caption{Stored Variables.}
\vspace{0.05in}
\centering
\begin{tabular}{cccc}
\toprule
\textbf{Database-1}	& \textbf{Database-2}	& \textbf{Database-3} & \textbf{Database-4} \\
\midrule
$a_1$		& $b_1$			& $c_1$		&$a_1+b_1+c_1$\\
$a_2$		& $b_2$			& $c_2$		&$a_2+b_2+c_2$\\
\bottomrule
\end{tabular}
\end{table}

The MDS storage criterion is easily verified, i.e., we may recover both messages from the storage of any 3 databases. 

{\bf PIR Code:} When we retrieve $W^1$, the answers are shown in the following table.
\vspace{-0.15in}

\begin{table}[H]
\caption{Answers for $W^1$.}
\vspace{0.05in}
\centering
\begin{tabular}{ccccc}
\toprule
$\mathsf{F}$ & \textbf{Database-1}	& \textbf{Database-2}	& \textbf{Database-3} & \textbf{Database-4} \\
\midrule
1& $a_2$		& $b_1$			& $c_1$		&$a_1+b_1+c_1$\\
2& $a_1$		& $b_2$			& $c_2$		&$a_2+b_2+c_2$\\
\bottomrule
\end{tabular}
\end{table}

When we retrieve $W^2$ or $W^3$, the answers are shown in the following tables.
\vspace{-0.15in}
\begin{table}[H]
\caption{Answers for $W^2$.}
\vspace{0.05in}
\centering
\begin{tabular}{ccccc}
\toprule
$\mathsf{F}$ & \textbf{Database-1}	& \textbf{Database-2}	& \textbf{Database-3} & \textbf{Database-4} \\
\midrule
1& $a_1$		& $b_2$			& $c_1$		&$a_1+b_1+c_1$\\
2& $a_2$		& $b_1$			& $c_2$		&$a_2+b_2+c_2$\\
\bottomrule
\end{tabular}
\end{table}

\begin{table}[H]
\caption{Answers for $W^3$.}
\vspace{0.05in}
\centering
\begin{tabular}{ccccc}
\toprule
$\mathsf{F}$ & \textbf{Database-1}	& \textbf{Database-2}	& \textbf{Database-3} & \textbf{Database-4} \\
\midrule
1& $a_1$		& $b_1$			& $c_2$		&$a_1+b_1+c_1$\\
2& $a_2$		& $b_2$			& $c_1$		&$a_2+b_2+c_2$\\
\bottomrule
\end{tabular}
\end{table}

{\bf Correctness and Privacy:} Both correctness and privacy are easy to see. 

{\bf Rate that outperforms separate MDS-PIR capacity:} The rate achieved is $\frac{L}{\sum_n l_n} = \frac{2}{4} = \frac{1}{2} \leq C(3,4,3)$, which is strictly higher than $C_\perp(3,4,3) = (1+\frac{3}{4} + (\frac{3}{4})^2)^{-1} = \frac{16}{37}$, the capacity of separate MDS storage code.

\subsection{General Proof: $(K,N,T) = (K, K+1, K), K \geq 2$}
The proof is a simple generalization of the example presented above. We set $L = 2$, and each message consists of $2$ bits from $\mathbb{F}_2$. Denote $W^k = (W^{k}_1; W^{k}_2), k \in \{1,\cdots,K\}$.

{\bf Storage Code:} Each database stores $\frac{LK}{T} = 2$ bits, and the stored variables are specified in the following table. Note that $K=T=N-1$.
\vspace{-0.15in}

\begin{table}[H]
\caption{Stored Variables.}
\vspace{0.05in}
\centering
\begin{tabular}{ccccc}
\toprule
\textbf{Database-1}	& \textbf{Database-2}	& $\cdots$ & \textbf{Database-$(N-1)$} & \textbf{Database-$N$} \\
\midrule
$W^{1}_1$		& $W^{2}_1$			& $\cdots$		&$W^{K}_1$ &$\sum_{k=1}^K W^{k}_1$\\
$W^{1}_2$		& $W^{2}_2$			& $\cdots$		&$W^{K}_2$  &$\sum_{k=1}^K W^{k}_2$\\
\bottomrule
\end{tabular}
\end{table}

The MDS storage criterion is easily verified, i.e., we may recover both messages from the storage of any $T=N-1$ databases. 

{\bf PIR Code:} When we retrieve $W^k$, the answers are shown in the following table.
\vspace{-0.15in}

\begin{table}[H]
\caption{Answers for $W^k$.}
\vspace{0.05in}
\centering
\begin{tabular}{cccccccc}
\toprule
$\mathsf{F}$ & \textbf{Database-1}	&$\cdots$  & \textbf{Database-$k$}	&$\cdots$ & \textbf{Database-$(N-1)$} & \textbf{Database-$N$} \\
\midrule
1& $W^{1}_1$				&$\cdots$	& $W^{k}_2$	&$\cdots$	&$W^{K}_1$ &$\sum_{k=1}^K W^{k}_1$\\
2& $W^{1}_2$				&$\cdots$	& $W^{k}_1$	&$\cdots$	&$W^{K}_2$ &$\sum_{k=1}^K W^{k}_2$\\
\bottomrule
\end{tabular}
\end{table}

{\bf Correctness and Privacy:} Follow immediately.

{\bf Rate that outperforms separate MDS-PIR capacity:} The rate achieved is $\frac{L}{\sum_n l_n} = \frac{2}{N} \leq C(K,K+1,K)$, while the capacity of separate MDS storage code is $C_\perp (K, K+1,K) = (1+\frac{N-1}{N} + \cdots + (\frac{N-1}{N})^{N-1} )^{-1} = \frac{1-\frac{N-1}{N}}{1-(\frac{N-1}{N})^{N}} = \frac{1}{N(1 - (\frac{N-1}{N})^N)}$. To prove $C(K,K+1,K) > C_\perp(K,K+1,K)$, it remains to show that
\begin{eqnarray}
\frac{2}{N} &>& \frac{1}{N(1 - (\frac{N-1}{N})^N)} \\
\Leftrightarrow (1-\frac{1}{N})^N &<& \frac{1}{2} \\
\Leftarrow (1-\frac{1}{N})^N &\leq& \frac{1}{e} < \frac{1}{2}.
\end{eqnarray}
The proof is thus complete.

\section{Regime Expansion Building upon Base Codes}
\label{sec:expansion}
We show that the two classes of base codes presented in previous sections for $(K,N,T)$ systems can be extended to $(K, mN, mT)$ systems ($m$ is a positive integer). 
We present this result in the next two subsections, one for each class of base codes. Let us start from the simpler case of $(K,K+1,K)$ systems.

\subsection{From $(K,K+1,K)$ to $(K,m(K+1),mK)$ Systems}
We show that $C(K, m(K+1), mK) > C_\perp(K,m(K+1), mK)$, where $K \geq 2$ and $m$ is a positive integer. 

The key idea is that we may split the messages and databases into $m$ generic copies so that the same PIR rate is preserved. 
Note that the separate MDS-PIR capacity is a function of $\frac{T}{N}$, i.e., $C_\perp(K,m(K+1),mK) = C_\perp(K, K+1,K)$. As $C(K,K+1,K) > C_\perp(K,K+1,K)$, it suffices to provide a joint MDS storage code for a $(K, m(K+1), mK)$ system that achieves the same PIR rate as that of a $(K, K+1, K)$ system (i.e., rate $\frac{2}{K+1}$). Such a storage and PIR code construction is presented next.

Each message is ``multiplied'' by $m$ so that we set $L = 2m$, and each message consists of $2m$ symbols from $\mathbb{F}_q$, where $q$ is an integer power of a prime number and is no fewer than $(m+1)K $. To highlight that the message symbols form two segments, we denote $W^k = ({\bf W}^k_1; {\bf W}^k_2) \in \mathbb{F}_q^{2 \times m}$, where ${\bf W}^k_i = (W^k_{i,1}, \cdots, W^k_{i,m}) \in \mathbb{F}_q^{1\times m}, i \in \{1,2\}$. 

{\bf Storage Code:} Each database stores $\frac{LK}{T} = 2$ symbols, as specified in the following table. For ease of presentation, the $N=m(K+1)$ databases are divided into $K+1$ groups ($m$ databases each) and labelled as ${DB(1,1)}, \cdots, {DB(1,m)}, \cdots, {DB}(K+1,m)$. Denote a group of databases as ${\bf DB}(k, :)$ $= \big(DB(k,1), \cdots, DB(k,m)\big), k \in \{1,2,\cdots,K+1\}$. A database in group $k, k\in\{1,\cdots,K\}$ stores 2 distinct $W^k$ symbols (one from ${\bf W}^k_1$ and one from ${\bf W}^k_2$). The $(K+1)$-th group of databases store generic combinations of the message symbols. Denote ${\bf W}_1 = (({\bf W}^1_1)^{T} ; \cdots; ({\bf W}^K_1)^{T}) \in \mathbb{F}_q^{mK\times 1}, {\bf W}_2 = ( ({\bf W}^1_2)^{T} ; \cdots; ({\bf W}^K_2)^{T} ) \in \mathbb{F}_q^{mK\times 1}$. ${\bf C}(i,:) \in \mathbb{F}_q^{1\times mK}, i \in \{1,\cdots,m\}$ denotes the $i$-th row of an $m \times mK$ Cauchy matrix ${\bf C}$ with elements ${\bf C}(i,j)$ in the form
\begin{eqnarray}
{\bf C}(i,j) = \frac{1}{\alpha_i - \beta_j},  \alpha_i \neq \beta_j, \forall i \in \{1,\cdots,m\}, j \in \{1,\cdots,mK\}.
\end{eqnarray}
Note that $q \geq (m+1)K$, therefore, such distinct $\alpha_i$'s and $\beta_j$'s exist.

\vspace{-0.15in}
\begin{table}[H]
\caption{Stored Variables.}
\vspace{0.05in}
\centering
\begin{tabular}{cccccccc}
\toprule
\textbf{${\bf DB}(1, :)$}	& \textbf{${\bf DB}(2, :)$}	& $\cdots$ &  \textbf{${\bf DB}(K, :)$} & \textbf{${\bf DB}(K+1, 1)$} & $\cdots$ & \textbf{${\bf DB}(K+1, m)$} \\
\midrule
${\bf W}^{1}_1$		& ${\bf W}^{2}_1$			& $\cdots$		&${\bf W}^{K}_1$ &${\bf C}(1,:) {\bf W}_1$ & $\cdots$ & ${\bf C}(m,:) {\bf W}_1$\\
${\bf W}^{1}_2$		& ${\bf W}^{2}_2$			& $\cdots$		&${\bf W}^{K}_2$  &${\bf C}(1,:) {\bf W}_2$ & $\cdots$ & ${\bf C}(m,:) {\bf W}_2$\\
\bottomrule
\end{tabular}
\end{table}

We now verify that the MDS storage criterion is satisfied, i.e., both messages can be recovered from the storage of any $T=mK$ databases. The two message segments ${\bf W}_1, {\bf W}_2$ are encoded in the same manner, so it suffices to consider one segment, say segment $1, {\bf W}_1$. Suppose among the $T=mK$ databases, $T_1 \leq (m-1)K$ databases are from the first $K$ database groups and the remaining $T - T_1$ databases are from the $(K+1)$-th database group. The $T_1$ databases from the first $K$ database groups contribute $T_1$ raw message symbols from ${\bf W}_1$, then we only need to show that the remaining $T - T_1$ symbols from ${\bf W}_1$ can be recovered from the $T - T_1$ databases of the $(K+1)$-th database group. This is equivalent to prove that a $(T - T_1) \times (T - T_1)$ sub-matrix of the Cauchy matrix ${\bf C} \in \mathbb{F}_q^{m\times mK}$ has full rank, which trivially holds for any Cauchy matrix.

{\bf PIR Code:} When we retrieve $W^k$, the answers are shown in the following table.
\vspace{-0.15in}
\begin{table}[H]
\caption{Answers for $W^k$.}
\vspace{0.05in}
\centering
\begin{tabular}{cccccccccc}
\toprule
$\mathsf{F}$& \textbf{${\bf DB}(1, :)$}	&$\cdots$& \textbf{${\bf DB}(k, :)$}	& $\cdots$ &  \textbf{${\bf DB}(K, :)$} & \textbf{${\bf DB}(K+1, 1)$} & $\cdots$ & \textbf{${\bf DB}(K+1, m)$} \\
\midrule
1 & ${\bf W}^{1}_1$		&$\cdots$ & ${\bf W}^{k}_2$			& $\cdots$		&${\bf W}^{K}_1$ &${\bf C}(1,:) {\bf W}_1$ & $\cdots$ & ${\bf C}(m,:) {\bf W}_1$\\
2 & ${\bf W}^{1}_2$		&$\cdots$& ${\bf W}^{k}_1$			& $\cdots$		&${\bf W}^{K}_2$  &${\bf C}(1,:) {\bf W}_2$ & $\cdots$ & ${\bf C}(m,:) {\bf W}_2$\\
\bottomrule
\end{tabular}
\end{table}

{\bf Correctness and Privacy:} Privacy follows from the observation that no matter which message is desired, the answer from any database is equally likely to come from message segment 1 or 2. To see correctness, note that all non-desired message symbols appeared in answers from the $(K+1)$-th database group are directly downloaded thus can be cancelled. $m$ desired symbols are directly downloaded and the other $m$ desired symbols can be successfully recovered because the $m$ linear combinations of desired symbols downloaded from the $(K+1)$-th database group have full rank (note that ${\bf C} \in \mathbb{F}_q^{m\times mK}$ is a Cauchy matrix). The rate achieved is $\frac{2}{K+1}$ as $L = 2m$ and we have downloaded one symbol from each of the $m(K+1)$ databases.

\subsection{From $(2,N,2)$ to $(2, mN, 2m)$ Systems}

We show that $C(2,mN,2m) > C_\perp(2,mN,2m)$, where $N\geq 3$ and $m$ is a positive integer. Similar to the reasoning in the previous section, it suffices to provide a joint MDS storage code for a $(2, mN, 2m)$ system that achieves the PIR rate $\frac{N-1}{N}$ (same as that of a $(2,N,2)$ system from Section \ref{sec:kk}).
The idea is also based on splitting the messages and databases. Let us start with an example where $N=4, m=2$.
\subsubsection{Example: $N=4, m=2$}
The message size is multiplied by $m=2$ so that we set $L = m(N-1) = 6$ and each message consists of $6$ symbols from $\mathbb{F}_q$, where $q$ will be specified later. At this point, it is useful to view $q$ as a sufficiently large prime number. Denote $W^1 = ({\bf a}_0; {\bf a}_1; {\bf a}_2)$, where ${\bf a}_i = (a_i; a_i'), i \in \{0,1,2\}$ and $W^2 = ({\bf b}_0; {\bf b}_1; {\bf b}_2)$, where ${\bf b}_i = (b_i; b_i'), i \in \{0,1,2\}$.

{\bf Storage Code:} Each database stores $\frac{LK}{T} = 3$ symbols, as specified in the following table. Define
\begin{eqnarray}
{\bf h}_i = (h_i, h_i') \in \mathbb{F}_q^{1\times 2}, {\bf g}_i = (g_i, g_i') \in \mathbb{F}_q^{1\times 2}, i \in \{1,2,\cdots,12\}.
\end{eqnarray}
We will show that there exist feasible choices of ${\bf h}_i, {\bf g}_i$. Specifically, we may choose $h_i, h_i', g_i, g_i'$ i.i.d. and uniform over $\mathbb{F}_q$.

\vspace{-0.15in}
\begin{table}[H]
\caption{Stored Variables.}
\vspace{0.05in}
\centering
\begin{tabular}{ccccccccc}
\toprule
\textbf{\scriptsize(DB1, DB2)}	 & \textbf{\scriptsize(DB3, DB4)} & \textbf{\scriptsize(DB5, DB6)} & \textbf{\scriptsize(DB7, DB8)} \\
\midrule
$(a_0,a_0')$		 & $(b_0,b_0')$		& $({\bf h}_1 {\bf a}_1 + {\bf g}_1 {\bf b}_0,{\bf h}_2 {\bf a}_1 + {\bf g}_2 {\bf b}_0)$ &$({\bf h}_7 {\bf a}_2 + {\bf g}_7 {\bf b}_0,{\bf h}_8 {\bf a}_2 + {\bf g}_8 {\bf b}_0)$\\
$(a_1,a_1')$		 & $(b_1,b_1')$		& $({\bf h}_3 {\bf a}_2 + {\bf g}_3 {\bf b}_1,{\bf h}_4 {\bf a}_2 + {\bf g}_4 {\bf b}_1)$ &$({\bf h}_9 {\bf a}_0 + {\bf g}_9 {\bf b}_1,{\bf h}_{10} {\bf a}_0 + {\bf g}_{10} {\bf b}_1)$\\
$(a_2,a_2')$		& $(b_2,b_2')$		& $({\bf h}_5 {\bf a}_0 + {\bf g}_5 {\bf b}_2,{\bf h}_6 {\bf a}_0 + {\bf g}_6 {\bf b}_2)$		&$({\bf h}_{11} {\bf a}_1 + {\bf g}_{11} {\bf b}_2,{\bf h}_{12} {\bf a}_1 + {\bf g}_{12} {\bf b}_2)$ \\
\bottomrule
\end{tabular}
\end{table}

To verify the MDS storage criterion, we need to show that both messages can be recovered from the storage of any $4$ databases. The detailed proof is deferred to the general proof presented in the next section and we give a sketch here. Every 4 databases contribute 12 linear combinations on the 12 message symbols and this linear mapping is given by a $12\times 12$ matrix. We view its determinant polynomial as a function of variables $(h_i, h_i', g_i, g_i')$. As shown in the general proof, these determinant polynomials are not zero polynomials. Overall we have $\binom{8}{4}$ determinant polynomials and each polynomial has degree at most $12$. Consider the product of all such determinant polynomials, which is another polynomial with degree at most $12 \times \binom{8}{4}$. Therefore by Schwartz-Zippel lemma, if we set $q > 12 \times \binom{8}{4}$, then the probability that this product polynomial evaluates to 0 is non-zero. In other words, we have found a feasible choice of $(h_i, h_i', g_i, g_i')$ that guarantees the storage code satisfies the MDS criterion.

{\bf PIR Code:} The PIR code is almost identical to that when $m=1$. When we retrieve $W^1$, the answers are shown in the following table.
\vspace{-0.15in}
\begin{table}[H]
\caption{Answers for $W^1$.}
\vspace{0.05in}
\centering
\begin{tabular}{ccccccccc}
\toprule
\textbf{\scriptsize(DB1, DB2)}	 & \textbf{\scriptsize(DB3, DB4)} & \textbf{\scriptsize(DB5, DB6)} & \textbf{\scriptsize(DB7, DB8)} \\
\midrule
$(a_0,a_0')$		 & $(b_0,b_0')$		& $({\bf h}_1 {\bf a}_1 + {\bf g}_1 {\bf b}_0,{\bf h}_2 {\bf a}_1 + {\bf g}_2 {\bf b}_0)$ &$({\bf h}_7 {\bf a}_2 + {\bf g}_7 {\bf b}_0,{\bf h}_8 {\bf a}_2 + {\bf g}_8 {\bf b}_0)$\\
$(a_1,a_1')$		 & $(b_1,b_1')$		& $({\bf h}_3 {\bf a}_2 + {\bf g}_3 {\bf b}_1,{\bf h}_4 {\bf a}_2 + {\bf g}_4 {\bf b}_1)$ &$({\bf h}_9 {\bf a}_0 + {\bf g}_9 {\bf b}_1,{\bf h}_{10} {\bf a}_0 + {\bf g}_{10} {\bf b}_1)$\\
$(a_2,a_2')$		& $(b_2,b_2')$		& $({\bf h}_5 {\bf a}_0 + {\bf g}_5 {\bf b}_2,{\bf h}_6 {\bf a}_0 + {\bf g}_6 {\bf b}_2)$		&$({\bf h}_{11} {\bf a}_1 + {\bf g}_{11} {\bf b}_2,{\bf h}_{12} {\bf a}_1 + {\bf g}_{12} {\bf b}_2)$ \\
\bottomrule
\end{tabular}
\end{table}

When we retrieve $W^2$, the answers are shown in the following table.
\vspace{-0.15in}
\begin{table}[H]
\caption{Answers for $W^2$.}
\vspace{0.05in}
\centering
\begin{tabular}{ccccccccc}
\toprule
\textbf{\scriptsize(DB1, DB2)}	 & \textbf{\scriptsize(DB3, DB4)} & \textbf{\scriptsize(DB5, DB6)} & \textbf{\scriptsize(DB7, DB8)} \\
\midrule
$(a_0,a_0')$		 & $(b_0,b_0')$		& $({\bf h}_5 {\bf a}_0 + {\bf g}_5 {\bf b}_2,{\bf h}_6 {\bf a}_0 + {\bf g}_6 {\bf b}_2)$&$({\bf h}_9 {\bf a}_0 + {\bf g}_9 {\bf b}_1,{\bf h}_{10} {\bf a}_0 + {\bf g}_{10} {\bf b}_1)$\\
$(a_1,a_1')$		 & $(b_1,b_1')$		& $({\bf h}_1 {\bf a}_1 + {\bf g}_1 {\bf b}_0,{\bf h}_2 {\bf a}_1 + {\bf g}_2 {\bf b}_0)$  &$({\bf h}_{11} {\bf a}_1 + {\bf g}_{11} {\bf b}_2,{\bf h}_{12} {\bf a}_1 + {\bf g}_{12} {\bf b}_2)$\\
$(a_2,a_2')$		& $(b_2,b_2')$		& $({\bf h}_3 {\bf a}_2 + {\bf g}_3 {\bf b}_1,{\bf h}_4 {\bf a}_2 + {\bf g}_4 {\bf b}_1)$		&$({\bf h}_7 {\bf a}_2 + {\bf g}_7 {\bf b}_0,{\bf h}_8 {\bf a}_2 + {\bf g}_8 {\bf b}_0)$ \\
\bottomrule
\end{tabular}
\end{table}

{\bf Correctness and Privacy:} Privacy is easily seen. To prove correctness, note that non-desired symbols can be cancelled and we only need to ensure the received desired equations are invertible to the message symbols. This claim follows from Schwartz-Zippel lemma that shows $({\bf h}_{2i-1}; {\bf h}_{2i}) \in \mathbb{F}_q^{2\times 2}, ({\bf g}_{2i-1}; {\bf g}_{2i}) \in \mathbb{F}_q^{2\times 2}$ have full rank with non-zero probability over a sufficiently large field. Here we have 12 matrices, each of which has dimension $2\times 2$ and has a determinant polynomial of degree at most 2.

Overall, we need to guarantee correctness and MDS criterion are simultaneously satisfied. Take the product of all determinant polynomials, whose degree is at most $12 \times \binom{8}{4} + 12 \times 2$. So we set $q > 12 \times \binom{8}{4} + 12 \times 2$ and by Schwartz-Zippel lemma, there exist a feasible choice of $(h_i, h_i', g_i, g_i')$ over $\mathbb{F}_q$.

\subsubsection{General proof: arbitrary $N\geq 3, m \geq 2$}
We set $L = m(N-1)$ and each message consists of $L$ symbols from $\mathbb{F}_q$, where $q$ is an integer power of a prime number and is no fewer than $2m(N-2)(N-1) + 2m(N-1)\binom{mN}{2m}$. Denote $W^1 = ({\bf a}_0; \cdots; {\bf a}_{N-2}) \in \mathbb{F}_q^{m(N-1) \times 1}$, where ${\bf a}_i = (a_{i,1}; \cdots; a_{i,m})\in \mathbb{F}^{m\times 1}_q, i \in \{0,1,\cdots,N-2\}$ and $W^2 = ({\bf b}_0; \cdots; {\bf b}_{N-2}) \in \mathbb{F}_q^{m(N-1) \times 1}$, where ${\bf b}_i = (b_{i,1}; \cdots; b_{i,m}) \in \mathbb{F}^{m\times 1}_q, i \in \{0,1,\cdots,N-2\}$.

{\bf Storage Code:} Each database stores $\frac{LK}{T} = N-1$ symbols. Denote the $mN$ databases as $DB(1,1), \cdots, DB(1,m), \cdots, DB(N,m)$. The stored variables $S_{n,j} \in \mathbb{F}_{q}^{(N-1) \times 1}, n \in \{1,\cdots,N\}, j \in \{1,\cdots,m\}$ are set as follows.

Denote $\overline{i} = i \mod (N-1)$. 
For any $j \in \{1,\cdots,m\}$,
\begin{eqnarray}
S_{1,j} &=& (S_{1,j,0}; \cdots; S_{1,j,N-2}) = (a_{0,j}; a_{1,j}; \cdots; a_{N-2, j}) \\
S_{2,j} &=& (S_{2,j,0}; \cdots; S_{2,j,N-2})  = (b_{0,j}; b_{1,j}; \cdots; b_{N-2, j}) \\
S_{n,j} &=& (S_{n,j,0}; \cdots; S_{n,j,N-2}) = ({\bf h}_{n,j,0} {\bf a}_{\overline{n-2}} + {\bf g}_{n,j,0} {\bf b}_{0} ; \cdots; {\bf h}_{n,j,N-2} {\bf a}_{\overline{n+N-4}} + {\bf g}_{n,j,N-2} {\bf b}_{N-2}),\nonumber\\
&&\qquad\qquad\qquad\qquad\qquad\qquad\qquad\qquad\qquad\qquad\qquad\qquad\qquad\qquad n \in \{3,\cdots,N\}
\end{eqnarray}
where for any $i \in \{0,\cdots,N-2\}$,
\begin{eqnarray}
{\bf h}_{n,j,i} \in \mathbb{F}_q^{1\times m}, {\bf g}_{n,j,i}  \in \mathbb{F}_q^{1\times m}.
\end{eqnarray}
The proof that there exist choices of ${\bf h}_{n,j,i}, {\bf g}_{n,j,i}$ such that the above storage code satisfies the MDS criterion is deferred to Section \ref{sec:gmdsproof}.

{\bf PIR Code:} When we retrieve $W^1$, we download one symbol from each database and the answers are set as follows. $\mathsf{F}$ is uniform over $\{0,1,\cdots,N-2\}$. When $\mathsf{F} = f \in \{0,1,\cdots, N-2\}$, for any $j \in \{1,2,\cdots,m\}$ we set
\begin{eqnarray}
&& A_{1,j}^{[1]} = S_{1,j,f}  = a_{f,j} \\
&& A_{2,j}^{[1]} = S_{2,j,f} = b_{f,j} \\
&& A_{n,j}^{[1]} =  S_{n,j,f} =  {\bf h}_{n,j,f} {\bf a}_{\overline{f+n-2}} + {\bf g}_{n,j,f} {\bf b}_f, n \in \{3,\cdots,N\}
\end{eqnarray}

When we retrieve $W^2$, the answers are set as follows. $\mathsf{F}$ is uniform over $\{0,1,\cdots,N-2\}$. When $\mathsf{F} = f \in \{0,1,\cdots, N-2\}$, for any $j \in \{1,2,\cdots,m\}$ we set
\begin{eqnarray}
&& A_{1,j}^{[2]} = S_{1,j,f} = a_{f,j} \\
&& A_{2,j}^{[2]} = S_{2,j,f} = b_{f,j} \\
&& A_{n,j}^{[2]} = S_{n,j,\overline{f-(n-2)}} =  {\bf h}_{n,j,\overline{f-(n-2)}}  {\bf a}_{f} + {\bf g}_{n,j,\overline{f-(n-2)}} {\bf b}_{\overline{f-(n-2)}}, n \in \{3,\cdots,N\}
\end{eqnarray}

{\bf Correctness and Privacy:} Privacy is easy to verify. For any $n,j$, $A_{n,j}^{[1]}$ and $A_{n,j}^{[2]}$ are identically distributed due to the modulo operation. Next consider correctness. Due to symmetry, we only need to consider the case when $W^1$ is the desired message. From $A_{2,j}^{[1]}, \forall j \in \{1,\cdots,m\}$, we have obtained all non-desired symbols $(b_{f,1}; \cdots; b_{f,m}) = {\bf b}_f$. After canceling the contribution of ${\bf b}_f$ from $A_{n,j}^{[1]}, n \geq 3$, we need to show that for any $n, f$, the $m \times m$ matrix $({\bf h}_{n,1,f}; \cdots; {\bf h}_{n,m,f})$ have full rank, which follows from Schwartz-Zippel lemma over a sufficiently large field. We have $2(N-2)(N-1)$ such matrices, each of size $m\times m$. The product of all these determinant polynomials has degree at most $2m(N-2)(N-1)$.

\subsubsection{Proof of MDS storage criterion}\label{sec:gmdsproof}
We show that when each element of ${\bf h}_{n,j,i}, {\bf g}_{n,j,i}$ is drawn independently and uniformly from $\mathbb{F}_q$, the probability that the MDS criterion is satisfied is non-zero so that there exists a feasible choice.

Consider any $T = 2m$ databases. We show that there exists an assignment of ${\bf h}_{n,j,i}, {\bf g}_{n,j,i}$ so that the mapping from the storage of the $T$ databases to the $2L$ message symbols is invertible. This shows that the $2L \times 2L$ matrix that describes the linear mapping has a non-zero determinant polynomial. Consider all choices of $\binom{mN}{2m}$ databases and take the product of all such determinant polynomials. Each polynomial has degree at most $2L$ so the degree of the product polynomial is at most $2L\binom{mN}{2m}$. Therefore, over a sufficiently large field, 
by Schwartz-Zippel lemma there exists a choice of ${\bf h}_{n,j,i}, {\bf g}_{n,j,i}$ so that all polynomials evaluate to non-zero values and the storage code is indeed MDS.

We are left to show that for any $T = 2m$ databases, we may assign ${\bf h}_{n,j,i}, {\bf g}_{n,j,i}$ (for a given choice of $T=2m$ databases) so that the storage is able to recover all $2L$ message symbols. The proof is based on a crucial property, stated in the following lemma. Define $\vec{a}_{j} = (a_{0,j}; a_{1,j}; \cdots; a_{N-2, j}), \vec{b}_{j} = (b_{0,j}; b_{1,j}; \cdots; b_{N-2, j}), j \in \{1,\cdots,m\}$.

\begin{Lemma}\label{lemma:any}
Consider any $n \in \{3,\cdots,N\}, j \in \{1,\cdots,m\}$, there exists a choice of ${\bf h}_{n,j,i}, {\bf g}_{n,j,i}, i \in \{0,1,\cdots,N-2\}$ so that from $S_{n,j}$, we may obtain $\vec{a}_{j^*}$ for any $j^* \in \{1,\cdots,m\}$ and another choice of ${\bf h}_{n,j,i}, {\bf g}_{n,j,i}, i \in \{0,1,\cdots,N-2\}$ so that from $S_{n,j}$, we may obtain $\vec{b}_{j^*}$ for any $j^* \in \{1,\cdots,m\}$.
\end{Lemma}
\begin{proof}[Proof of Lemma \ref{lemma:any}] The proof is fairly simple because $S_{n,j}$ contains all symbols from $\vec{a}_{j^*}$ and $\vec{b}_{j^*}$ for any $j^*$.
Consider first $S_{n,j} = \vec{a}_{j^*}$. For all $i \in \{0,1,\cdots,N-2\}$, set
\begin{eqnarray}
{\bf g}_{n,j,i} &=& {\bf 0}\\
{\bf h}_{n,j,i} &=& {\bf e}_{j^*} 
\end{eqnarray}
where ${\bf e}_{j^*}$ is a $1\times m$ unit vector so that only the element of the $j^*$-th position is 1 and all other elements are 0,
then we have
\begin{eqnarray}
S_{n,j} &=& (S_{n,j,0}; \cdots; S_{n,j,N-2}) = ({\bf h}_{n,j,0} {\bf a}_{\overline{n-2}} + {\bf g}_{n,j,0} {\bf b}_{0} ; \cdots; {\bf h}_{n,j,N-2} {\bf a}_{\overline{n+N-4}} + {\bf g}_{n,j,N-2} {\bf b}_{N-2}) \notag\\
\\
&=&  ({\bf h}_{n,j,0} {\bf a}_{\overline{n-2}} ; \cdots; {\bf h}_{n,j,N-2} {\bf a}_{\overline{n+N-4}} ) \\
&=&  ({a}_{\overline{n-2},j^*}; \cdots; {a}_{\overline{n+N-4}, j^*} )
\end{eqnarray}
which is a cyclic shift of $\vec{a}_{j^*} = (a_{0,j^*}; a_{1,j^*}; \cdots; a_{N-2, j^*})$.

The case of $S_{n,j} = \vec{b}_{j^*}$ follows similarly from the assignment given above.
\end{proof}

Fix any $T = 2m$ databases. Suppose $T_1 \leq 2m$ databases are from $DB(i,j)$ where $i \in \{1,2\}, j \in \{1,\cdots,m\}$ and they will contribute $T_1$ distinct $\vec{a}_{j_1^*}$ and $\vec{b}_{j_2^*}$ vectors. The remaining $T - T_1$ databases are from $DB(n,j)$ where $n \in \{3,\cdots,N\}, j \in \{1,\cdots,m\}$ and our goal is to recover all remaining $T-T_1$ $\vec{a}_{j_3^*}$ and $\vec{b}_{j_4^*}$ vectors. 
{\color{black} We can identify a one-to-one mapping between the $T-T_1$ databases and the remaining $(T-T_1)$ $\vec{a}_{j_3^*}$ and $\vec{b}_{j_4^*}$ vectors, and apply Lemma \ref{lemma:any} to find the assignment such that the $\vec{a}_{j_3^*}$ and $\vec{b}_{j_4^*}$ vectors are fully recovered.} 
Hence from any $T$ databases, we may recover $(\vec{a}_1,\cdots,\vec{a}_{m})$ and $(\vec{b}_1,\cdots,\vec{b}_{m})$, i.e., all symbols from $W^1$ and $W^2$. Therefore, there indeed exists a choice of ${\bf h}_{n,j,i}, {\bf g}_{n,j,i}$ for which the determinant polynomial is not zero. 

Finally, we need to consider correctness and MDS criterion jointly and show that there exist a single choice of ${\bf h}_{n,j,i}, {\bf g}_{n,j,i}$ that satisfies both constraints at the same time. The product of all determinant polynomials has degree at most $2m(N-2)(N-1) + 2m(N-1)\binom{mN}{2m}$ and as $q > 2m(N-2)(N-1) + 2m(N-1)\binom{mN}{2m}$, Schwartz-Zippel lemma guarantees the existence of a feasible choice.

\section{Conclusion}
\label{sec:conclusion}

We considered the problem of private information retrieval from MDS-coded databases. Different from the prevailing approach in the literature where the messages are encoded separately using MDS codes, we consider encoding and storing the messages jointly using an MDS code into the databases. There are many cases for which by jointly MDS-coding, we can break the capacity barrier of the separate coding MDS-PIR. To establish this result, two novel code constructions and the corresponding PIR protocols are presented, and moreover, an expansion technique is introduced to allow more general  parameters. {\color{black} The capacity of PIR with joint MDS storage, especially the converse side, remains an interesting future direction.}

\let\url\nolinkurl
\bibliographystyle{IEEEtran}
\bibliography{Thesis,PIR}
\end{document}